\documentclass[sigconf]{acmart}
\usepackage{booktabs} %
\usepackage{amsmath}
\usepackage{amsthm}
\usepackage{amssymb}
\usepackage{enumerate}
\usepackage{algorithm}
\usepackage{algpseudocode}
\usepackage{xspace}
\usepackage{multirow,multicol}
\usepackage{framed}
\usepackage{balance}
\newif\ifpreprint
\preprintfalse

\usepackage[shortlabels]{enumitem}

\usepackage{calc}

\newenvironment{myquote}%
  {\smallskip \list{}{\leftmargin=0.2in\rightmargin=0.2in}\item[] \em}%
  {\endlist}

\newcommand{\pathoram}{{Pa\-th-OR\-AM\xspace}}
\newcommand{\travisworam}{{Hi\-VE-Wo\-ORAM\xspace}}
\newcommand{\ourworam}{{Det\-Wo\-ORAM\xspace}}
\usepackage{etoolbox}
\patchcmd{\paragraph}{\@parfont}{\bfseries}{}{} 
\patchcmd{\paragraph}{\parindent}{0pt}{}{} 

\newcommand{\BT}{\begin{theorem}}
\newcommand{\ET}{\end{theorem}}

\newcommand{\BD}{\begin{definition}}
\newcommand{\ED}{\end{definition}}

\newcommand{\BCR}{\begin{corollary}}
\newcommand{\ECR}{\end{corollary}}

\newcommand{\BEX}{\begin{example}}
\newcommand{\EEX}{\end{example}}

\newcommand{\BL}{\begin{lemma}}
\newcommand{\EL}{\end{lemma}}

\newcommand{\BP}{\begin{proposition}}
\newcommand{\EP}{\end{proposition}}

\newcommand{\BCM}{\begin{claim}}
\newcommand{\ECM}{\end{claim}}

\newcommand{\BPF}{\begin{proof}}
\newcommand{\EPF}{\end{proof}}
\newcommand{\BEN}{\begin{enumerate}}
\newcommand{\EEN}{\end{enumerate}}

\newcommand{\BI}{\begin{itemize}}
\newcommand{\EI}{\end{itemize}}

\newcommand{\BO}{\begin{observation}}
\newcommand{\EO}{\end{observation}}

\newcommand{\BDS}{\begin{description}}
\newcommand{\EDS}{\end{description}}

\def\FF{{\mathbb F}}       %
\def\zo{{\{0,1\}}}    %

\def\from{{\,\leftarrow\,}}

\newcommand{\ceil}[1]{{\left\lceil{#1}\right\rceil}}
\newcommand{\floor}[1]{{\left\lfloor{#1}\right\rfloor}}

\def\squareforqed{\hbox{\(\blacksquare\)}}
\def\qed{\hfill \squareforqed}

\newcommand{\secparam}{\ensuremath{\lambda}}

\newcommand{\ignore}[1]{}

\def\dec{\textsc{dec}}

\newcommand{\negl}{{\sf negl}}

\newcounter{defcounter}
\newlength{\protowidth}

\newcommand{\olrk}[1]{%
   \ifx\nursymbol#1\else\!\!\mskip4.5mu plus 0.5mu\left(#1\right)\fi}
\newcommand{\elrk}[1]{%
   \ifx\nursymbol#1\else%
        \!\!\mskip4.5mu plus0.5mu\left[\mskip2.5mu plus0.5mu #1\right]\fi}

\def\stash{\ensuremath{{\sf stash}}}
\def\path{\ensuremath{{\sf location}}}

\def\Access{\mathsf{PhysicalAcc}}
\def\Filter{\mathsf{WOnly}}

\usepackage{graphicx}
\graphicspath{{./figures/}}

\usepackage{xcolor}

\usepackage{tikz}
\usetikzlibrary{calc,positioning,shapes,decorations.pathreplacing,shadows}

\usepackage{xspace}
\setcopyright{rightsretained}

\usepackage[utf8]{inputenc}

\usepackage{enumitem}

\numberwithin{equation}{section}

\copyrightyear{2017}
\acmYear{2017}
\setcopyright{usgov}
\acmConference{CCS'17}{}{Oct. 30--Nov. 3, 2017, Dallas, TX, USA.}
\acmPrice{15.00}
\acmDOI{10.1145/3133956.3134051}
\acmISBN{978-1-4503-4946-8/17/10}

\title{Deterministic, Stash-Free Write-Only ORAM}

\newcommand{\usna}{\affiliation{
  \institution{United States Naval Academy}
  \city{Annapolis}
  \state{Maryland}
  \country{U.S.A.}
}}

\author{Daniel S. Roche}
\orcid{0000-0003-1408-6872}
\usna
\email{roche@usna.edu}

\author{Adam Aviv}
\usna
\email{aviv@usna.edu}

\author{Seung Geol Choi}
\usna
\email{choi@usna.edu}

\author{Travis Mayberry}
\usna
\email{mayberry@usna.edu}

\begin{document}

\begin{abstract}

  Write-Only Oblivious RAM (WoORAM) protocols provide privacy by
  encrypting the contents of data and also hiding the pattern of {\em
  write operations} over that data.  WoORAMs provide better privacy
  than plain encryption and better performance than more general ORAM
  schemes (which hide both writing \emph{and} reading access patterns),
  and the write-oblivious setting has been applied to important
  applications of cloud storage synchronization and encrypted hidden
  volumes.  In this paper, we introduce an entirely new technique for
  Write-Only ORAM, called \ourworam{}.  Unlike previous solutions,
  \ourworam{} uses a deterministic, sequential writing pattern
  without the need for any ``stashing'' of blocks in local
  state when writes fail. Our protocol, while conceptually
  simple, provides substantial improvement over prior solutions, both
  asymptotically and experimentally.  In particular, under typical
  settings the \ourworam{} writes only 2 blocks (sequentially) to
  backend memory for each block written to the device, which is
  optimal. We have implemented our solution using the BUSE (block
  device in user-space) module and tested \ourworam{} against both an
  encryption only baseline of dm-crypt and prior, randomized WoORAM
  solutions, measuring only a 3x-14x slowdown compared to an
  encryption-only baseline and around 6x-19x speedup compared to
  prior work.

\end{abstract}

\maketitle

\section{Introduction}

\paragraph{ORAM}
Even when data is fully encrypted, the sequence of \emph{which operations
have been performed} may be easily observed. This \emph{access pattern}
leakage is prevented by using Oblivious RAMs (or ORAMs), which are
protocols that allow a client to access files in storage without
revealing the sequence of operations over that data. ORAM solutions that
have been proposed provide strong privacy by guaranteeing that anyone
who observes the entire communication channel between client and backend
storage cannot distinguish any series of accesses from random.
Due to this strong privacy guarantee, ORAM has been used as a powerful tool in
various application settings such as secure cloud storage (e.g.,
\cite{NDSS:SteShi13,SP:SteShi13,NDSS:MayBlaCha14}),
secure multi-party computation (e.g.,
\cite{CCS:GKKKMR12,SP:LHSKH14,SP:LWNHS15,CCS:WanChaShi15,SP:ZWRGDE16}),
and secure processors (e.g., \cite{freecursive,ghostrider,hop}).

Unfortunately, in order to achieve this obliviousness, ORAM schemes
often require a substantial amount of shuffling during every access,
requiring more encrypted data to be transferred than just the data
being written/read. Even Path ORAM \cite{CCS:SDSFRY13}, one of the
most efficient ORAM constructions, has an $\Omega(\log N)$
multiplicative overhead in terms of communication complexity compared
to non-private storage.

\paragraph{WoORAM}
Write-only ORAM
(WoORAM)~\cite{EPRINT:LiDat13,CCS:BMNO14} introduces a relaxed
security notion as compared to ORAMs, where only the write pattern
needs to be oblivious. That is, we assume a setting in which
the adversary is able to see
the entire history of which blocks have been written to the backend, or
to view arbitrary snapshots of the storage, but the adversary cannot 
see which blocks are being read.

Every ORAM trivially satisfies the properties of WoORAM, but entirely
different (and possibly more efficient) WoORAM solutions are available
because
\emph{a WoORAM is by definition secure even if $read$s are not
oblivious}.
WoORAM schemes can be used in application settings where adversaries are unable
to gather information about physical $read$s. In such settings, the weaker
security guarantee of WoORAM still suffice to hide the access patterns from
the adversary of limited power.

{\em Deniable storage}~\cite{CCS:BMNO14} is one such application. In this
setting, a user has a single encrypted volume and may optionally have
a second, hidden volume, the existence of which the user wishes to be
able to plausibly deny. For example, a laptop or mobile device owner may
be forced to divulge their device decryption password at a border
crossing or elsewhere. The adversary may also be able to view multiple
snapshots of the disk, either at different times or through physical
forensic information remaining on the storage medium.
Even given \emph{every past state} of storage, an
adversary should not be able to guess whether the user has a second
hidden volume or not.
In this context, it is reasonable to assume that the
adversary won't get any information about block reads that have taken
place in the disk, since read operations do not usually leave traces on
the disk.  Based on this observation, a hidden volume encryption
(HiVE) for deniable storage was constructed based on
WoORAM~\cite{CCS:BMNO14}.

We proposed another example application in
\cite{oblivisync} for {\em synchronization based cloud
storage and backup services}. Here, the user holds the entire contents
of data locally, and uses a service such as DropBox to synchronize
with other devices or store backups.
The service provider or an eavesdropper on the network only observes
what the user writes to the synchronization folder, but does not see any
read operations as these are done locally without the need for network
communication.
We showed in \cite{oblivisync} showed that
WoORAMs can provide efficient protection in this scenario, as well as
protection against timing and file size distribution attacks.

In both cases, what makes WoORAMs attractive is that
they can achieve security much, much more efficiently than the full
read/write oblivious ORAMs such as \pathoram{}. For example, consider
storing $N$ size-$B$ data blocks in a non-recursive setting in which
the client has enough memory
to contain the entire position map of size $O(N\log N)$,
Blass~et~al.~\cite{CCS:BMNO14} provided a WoORAM construction
(hereafter, \travisworam) with {\em optimal} asymptotic communication
overhead of $O(B)$ and negligible stash overflow probability. As a
comparison, fully-functional read/write ORAM schemes --- again, even
without the position map --- have an
overhead of $\Omega(B \log N)$.

\paragraph{Towards better efficiency with realistic client memory}
Although \travisworam{} has a better asymptotic communication
complexity than \pathoram{} in the non-recursive setting (i.e., with
client memory of size $O(N \log N)$), the situation is different
when the size of the \emph{unsynchronized} client memory is smaller (i.e., polylogarithmic in
$N$). This could be because the client really has less memory, or
because the state needs to be synchronized frequently (as in a
multi-user setting).
In this case, the client cannot maintain the entire
position map in memory, and so the position map storage needs to be
outsourced to the server as another WoORAM. This encoding typically
occurs via a recursive process, storing the position map in
recursively smaller WoORAMS, until the final WoORAM is small enough to
fit within client memory.
Therefore, in the uniform block setting where every storage block has
the same size,
both \travisworam{} and \pathoram{} have the same overhead $O(B\log^2 N)$
with poly-logarithmic client memory size.\footnote{The multiplicative overhead
$O(\log^2 N)$ can be reduced to additive overhead of $O(\log^3 N)$ if the size
of the block can be non-uniform~\cite{CCS:SDSFRY13,CCS:BMNO14}. However,
throughout the paper we will consider the uniform block setting, since the two
use cases we consider above assume uniform block sizes. We note that our
construction still has better additive overhead of $O(\log^2 N)$ even in the
non-uniform block setting.}

Hence, we ask the following question:

\begin{myquote}
Can we achieve WoORAM with better asymptotic communication complexity in the
setting of polylogarithmic client memory and uniform blocks?
\end{myquote}

\begin{table*}[t]
\ifpreprint
\makebox[\textwidth][c]{
\else
\begin{center}
\resizebox{0.8\textwidth}{!}{%
\fi
\begin{tabular}{c|cc|cc|c|c|}
\hline
 & \multicolumn{2}{c|}{Logical Read} & \multicolumn{2}{c|}{Logical Write} & Unsynchronized & Security \\
 \cline{2-3}\cline{4-5}
 & Physical Read & Physical Write & Physical Read & Physical Write & Client Memory &~ \\
\hline
\pathoram{} \cite{CCS:SDSFRY13}
  & $O(B\log^2 N)$ & $O(B\log^2 N)$ & $O(B\log^2 N)$ & $O(B\log^2 N)$
  & $\omega(B \log N)$ & RW\\
\travisworam{} \cite{CCS:BMNO14}
  & $O(B\log N)$ & 0 & $O(B\log^2 N)$ & $O(B\log N)$
  & $\omega(B \log N)$ & W-only \\
\ourworam{}
  & $O(B\log N)$ & 0 & $O(B\log N)$ & $2B$
  & $O(B)$ & W-only \\
\hline
\end{tabular}
\ifpreprint
\else
}
\end{center}
\fi
\ifpreprint
}
\scriptsize
\fi
\centering
$B$ denotes the block size (in bits), and $N$ denotes the number
of logical blocks. We assume $B = \Omega(\log^2 N)$.
\caption{Communication complexity and client memory size for various ORAMs in the uniform block setting
\label{tab:orams}}
\end{table*}

\subsection{A Deterministic Approach to WoORAMs}
In answering the question above, observe that the security requirement
of WoORAMs are {\em much weaker} than that of ORAMs. Namely, only the
write operations need to be oblivious, and the read operations can
occur using different protocols than that of writing. This opens the
door to a {\em radically different approach} toward constructing a
WoORAM scheme.

\paragraph{Traditional approaches.}
Traditionally, in ORAM schemes as well as WoORAM, to write data $d$, the
oblivious algorithm selects $k$ blocks in some random process
storage in order to write. In \pathoram{}, those $k$ blocks form a
path in a tree, while in \travisworam{}, they are uniformly sampled from
a flat storage array of blocks.  All $k$ blocks are re-encrypted,
and the new block $d$ is inserted if any existing blocks are empty.

One of the challenges with this approach is that there is the
possibility for a write to fail if none of the random $k$ blocks are
empty and thus $d$ cannot be inserted. Instead $d$ is placed into a
{\em stash} in
reserved client memory until it may be successfully written to the
ORAM (or WoORAM) when two or more of the $k$ blocks are
empty. Fortunately, the probability of this event is bounded, and thus
the size of the stash can also be bounded with negligible stash
overflow probability. The schemes will, with overwhelming probability,
work for small client memory.

\paragraph{Main observations.}
After careful inspection of the security proofs, we discovered that
{\em random slots are mainly
used to hide read accesses, not write accesses!} That is, the
challenge for ORAMs is that successive reads of the same data must occur
in a randomly indistinguishable manner. For example, without the
technique of choosing random slots, two logical reads on the same
address may result in reading the same physical address twice, in
which case the read accesses are not oblivious.
In the WoORAM setting, however, {\em the scheme may still be secure even
if reads are not oblivious}, since the security requirement doesn't
care about physical reads!  Based on this observation, we ask:

\begin{myquote}
Can we construct a deterministic WoORAM scheme using a radically different framework?
\end{myquote}

\subsection{Our Work: \ourworam{}}
We answer both of the above questions affirmatively. In what
follows, we describe the main features and contributions of \ourworam{}.

\paragraph{Deterministic, sequential physical accesses}
\ourworam{} departs from the traditional approach in constructing a
WoORAM scheme in that the write pattern is deterministic and
sequential.  Roughly speaking, if some logical $write$ results in writing
the two physical blocks $i$ and $j$, the next logical $write$ will result in
writing in physical blocks $(i+1) \bmod N$ and $(j+1) \bmod M$, where
$M$ is a parameter in the system.

\paragraph{No stash}
The deterministic nature of the physical writes also implies that a
stash is no longer needed. A write will {\em always succeed} and always
occurs in a free block.  Therefore, we were able to {\em remove the
notion of stash completely} in our scheme. To elaborate this point,
we give a very simple toy construction that captures these aspects in
Section~\ref{sec:toy}.
Due to the deterministic access pattern and the absence of stash, security
analysis of our scheme is extremely simple.

\paragraph{{\em Optimal} communication complexity of physical writes}
Each logical read or write operation from the client's end results in
some physical reads and/or writes to backend memory. In the
\emph{uniform block setting}, we assume there is a block size $B$,
presumably dictated by the underlying medium, and that all reads and
writes must occur in multiples of $B$. The communication complexity is
then the total number of bytes transferred for a given operation, which
necessarily is a multiple of $B$.

\ourworam{} has better asymptotic communication complexity than previous
constructions (see Table~\ref{tab:orams}). In particular, \ourworam{}
improves the complexity of write operations compared to \travisworam{}
by a factor of $\log N$. Note that, even though read operations are
assumed to be hidden from an observer, the asymptotic cost of reads is
still very significant for practical performance.

We stress that to the best of our knowledge, \ourworam{} is {\em the first
WoORAM} construction whose physical write cost is $2B$ for a single logical
write. In other words, the physical-write overhead is {\em a single block
(additive!)}, which is {\em exactly optimal in the uniform block setting with
small client storage}.

\paragraph{Optimization techniques: ODS and packing}
We applied two optimization techniques to further reduce the
communication complexity and improve practical performance.

First, we created a new write-only oblivious data structure (ODS),
in the form of a Trie, to function as the position map.
As with previous tree-based ODS schemes~\cite{CCS:WNLCSS14,SP:RocAviCho16}, our
ODS scheme avoids recursive position map lookups by employing a {\em
pointer-based technique}. That is, the pointer to a child node in a Trie node
{\em directly points to a physical location} instead of a logical location, and
therefore it is no longer necessary to translate a logical address to a physical
address within the Trie itself. We note that the ODS idea has previously
been applied to WoORAM by \cite{PETS:ChaCheSio17}, although their overall
scheme turns out the be insecure (see Section~\ref{sec:attack}).

With the simpler position map stored as a Trie and the deterministic
write-access pattern in \ourworam{}, we can manipulate the parameters to
optimize the procedure with \ourworam{}. In particular, we will show how to
pack write-updates of the position map Trie into block size chunks.  With
additional interleaving techniques, we will show that we can achieve a minimal
communication complexity of $2B$, one block for the data and one block for
position map and other updates.  The details of these techniques are described
in Section~\ref{sec:trie} and~\ref{sec:packing}.

\paragraph{Stateless WoORAM}
WoORAM is usually considered in a single-client model, but it is sometimes
useful to have multiple clients accessing the same WoORAM.
In a multi-client setting, \emph{even if the client has a large amount
of local memory available}, our improvements in local storage of
eliminating stash and optimizing the position map are significant.

Because our scheme
has no stash, we can convert our scheme to a stateless WoORAM with no overhead
except for storing the encryption key and a few counter variables.
On the other hand, previous schemes such as \pathoram{} and \travisworam{} must
maintain the stash of size $\omega(B \log N)$, which in the stateless
setting must be transferred
in its entirety on each write operation in order to maintain
obliviousness.

\paragraph{Less randomness and storage for IVs}
The deterministic and sequential access pattern fits nicely with encryption of
each block using \emph{counter mode}. Suppose the previous writes so far have
cycled the physical storage $i$ times, and physical block $j$ is about to be
encrypted.  Then, the block can be encrypted with the counter $i \| j \|
0^\ell$, where $\ell$ depends on how many times one needs to apply a block
cipher to encrypt the entire block. That is, we can get
indistinguishable symmetric encryption by storing just a single IV.

We stress that the above optimization cannot be applied to previous schemes due
to the randomized procedure.  For example, \travisworam{} chooses instead
random IVs to encrypt each block.  These IVs must be stored separately on the
server adding to the communication cost overhead.

Additionally, we remark that the implementation~\cite{CCS:BMNO14}
of~\travisworam{} used RC4 as a PRG to choose random IVs for
performance reasons, but the insecurity of RC4 lead to an attack on the
implementation~\cite{ASIACCS:PatStr15}.

\paragraph{Implementation and experiments}
We have implemented \ourworam{} using C++ and BUSE ({\em block device
  in user space}). We tested our implementation using Bonnie++ and fio on both
a spinning-platter hard drive and a solid state drive, comparing the
implementation to a baseline that performs encryption only (no
obliviousness) as well as to an implementation of \travisworam{}.
We found that \ourworam{} incurs only a 3x-14x
slowdown compared to an encryption-only baseline and around 6x-19x
speedup compared to \travisworam{}.

\paragraph{Insecurity of other proposed WoORAM improvement}
DataLair~\cite{PETS:ChaCheSio17} is another
WoORAM scheme that has been proposed recently with the
goal of improving the practical performance compared to \travisworam{}.
As our secondary contribution, we analyzed this WoORAM protocol,
which achieves faster writes by tracking empty blocks within the
WoORAM.
We show in Section~\ref{sec:attack} that,
unfortunately, the construction does not satisfy
write-only obliviousness.

\section{Background}

\subsection{Write-only ORAM}
\label{sec:woramdef}
\paragraph{ORAM}
An Oblivious RAM (ORAM) allows a client to store and manipulate an array of $N$
blocks on an untrusted, honest-but-curious server without revealing the data or access patterns to
the server. Specifically, the (logical) array of $N$ blocks is indirectly
stored into a specialized backend data structure on the server, and an ORAM
scheme provides an access protocol that implements each logical access with a
sequence of physical accesses to the backend structure. An ORAM scheme is
secure if for any two sequences of logical accesses of the same length, the
physical accesses produced by the access protocol are computationally
indistinguishable.

More formally, let $\vec{y} = (y_1, y_2, \ldots)$ denote a sequence of
operations, where each $y_i$ is a $read(a_i)$ or a $write(a_i, d_i)$; here,
$a_i \in [0, N)$ denotes the logical address of the block being read or
written, and $d_i$ denotes a block of data being written. For an ORAM scheme
$\Pi$, let $\Access^\Pi(\vec{y})$ denote the physical access pattern that its
access protocol produces for the logical access sequence $\vec y$.  We say the
scheme $\Pi$ is {\em secure} if for any two sequences of operations
$\vec{x}$ and $\vec{y}$ of the same length, it holds
$$ \Access^\Pi(\vec{x}) ~~ \approx_c ~~ \Access^\Pi(\vec{y}), $$
where $\approx_c$ denotes computational
indistinguishability (with respect to the security parameter $\secparam$).

Since the seminal work by Goldreich and Ostrovsky~\cite{GO96},
many ORAM schemes have been proposed and studied in the literature; see
Section~\ref{sec:related} for more related work.

\paragraph{WoORAM}
Blass~et~al.~\cite{CCS:BMNO14} considered a relaxed security notion of
write-only ORAM (WoORAM), where only the $write$ physical accesses are required to
be indistinguishable. In particular, we say an ORAM scheme $\Pi$ is {\em
write-only oblivious} if for any two sequences of logical accesses $\vec{x}$
and $\vec{y}$ containing the same number of $write$ operations, it holds
$$\Filter(\Access^\Pi(\vec{x})) ~~
    \approx_c ~~ \Filter(\Access^\Pi(\vec{y})), $$
where $\Filter$ denotes a function that filters out the read physical accesses
but passes the write physical accesses.

They also gave an WoORAM construction which is much more efficient than full
ORAM constructions. We will briefly describe their construction below.

\subsection{\travisworam}
\paragraph{Setting}
In \cite{CCS:BMNO14}, to store $N$ logical blocks, the server needs a physical
array $D$ of $M \ge 2N$ elements, where each element is a pair $(a, d)$, where
$a$ is the logical address and $d$ is the actual data. Obviously, all the data
in the backend storage $D$ is encrypted with an IND-CPA encryption scheme;
throughout the paper, we will implicitly assume that the backend data is
encrypted with an IND-CPA encryption scheme, even if we don't use any
encryption notations.

The client maintains a buffer, called $\stash$, that temporarily holds the
blocks yet to be written to $D$.  We assume for now that the client also
maintains the position map $pos$ in its memory; the map $pos$ translates a
logical address into the corresponding physical address.

This protocol depends crucially on parameter $k$, the number of physical
writes per logical write. This is selected to ensure a very low
probability of filling up the stash;
according to~\cite{CCS:BMNO14}, for $k=3$, the
probability of having more than $50$ items in the stash at any given
time is bounded by
$2^{-64}$.

\paragraph{Write algorithm}
The access protocol for $write(a,d)$ works as follows.
\BEN
\item Insert $(a, d)$ into $\stash$

\item Choose $k$ physical addresses $r_1,\ldots,r_k$ uniformly at random
  from $[0,M)$.

\item For $i \in \{1,\ldots,k\}$ do:
    \BEN
    \item Determine whether $D[r_i]$ is free by checking whether
      $pos[D[r_i].a] \ne r_i$.

    \item If $D[r_i]$ is free and $\stash$ is nonempty,
    remove an element $(\alpha,\delta)$ from stash,
    set $D[r_i] \gets \delta$, $D[r_i].a \gets \alpha$, and
    update the position map
    $pos[\alpha] \gets r_i$.

    \item
    Otherwise, rewrite $D[r_i]$ under a new random IV.
    \EEN
\EEN

\paragraph{Communication complexity}
Let $M = O(N)$.  Without considering the position map, their access protocol for write has
fantastic communication complexity of $O(k(\log N + B))$, where $B$ is the size
of a data block. In particular, with $k=3$ and assuming $B = \Omega(\log N)$,
the communication complexity is $O(B)$.
However, the size of the position map is $\Omega(N \cdot \log N)$, which is usually
too large for the client to store in memory. This issue can be addressed by
recursively storing the map in smaller and smaller ORAMs. Taking these recursion
steps into account, the eventual communication complexity becomes $O(B \log^2 N)$.

\section{Deterministic WoORAM Design}
In this section, we describe the algorithm for \ourworam{}
construction. We begin by first describing a ``toy construction'' that
has some of the key properties as the final algorithm, such as not
employing a stash while using a deterministic write pattern. From this
toy construction, we make a series of improvements that lead to our
actual \ourworam{} construction with sequential write pattern and $2B$
communication cost per write.

\subsection{A {\em Toy} Deterministic WoORAM Construction} \label{sec:toy}
The {\em toy} deterministic WoORAM construction is inspired by the
{\em square-root} ORAM solution of Goldreich and Ostrovsky
\cite{GO96}, adapted to the write-only oblivious setting. For now,
we set aside the issue of the position map, which one could consider
being stored by the user locally and separately from the primary
procedure. Later, we will describe a method for storing the position
map within an adjacent, smaller WoORAM.

\paragraph{Toy Construction.}
Physical storage consists of an array $D$ of $2N$
data blocks of $B$ bits each.
$D$ is divided into two areas, a {\em main
area} and a {\em holding area} (see Figure~\ref{fig:toy}), where each
area contains $N$ blocks.

The key invariants of this construction, which will continue even with
our complete non-toy construction later, are:
\begin{itemize}[nosep]
  \item Every block in the main area is stored \emph{at its actual address};
    therefore the main area does not need any position map.
  \item Every block in the holding area is overwritten only after it has
    been copied to the main area.
\end{itemize}

\tikzset{
w/.style={draw,rectangle,text height=8pt,text width=100pt,align=center,fill=white},
g/.style={draw,rectangle,text height=8pt,text width=100pt,align=center,fill=gray!30},
}

\begin{figure}[h]
\small
\centering
\noindent\begin{tikzpicture}[node distance=0]
\def\lnode#1#2{%
\node[w, label=center:{#2}] (#1) {}}
\def\lnodex#1#2#3{%
\node[g, right = 0pt of #1, label=center:{#3}] (#2) {}}

\lnode{a}{main area};
\lnodex{a}{b}{holding area};

\draw[thick, <->] (-53pt,-16pt) -- (52pt,-16pt) node[below left = 2pt and 40pt] {$N$ blocks };
\draw[thick, <->] (53pt,-16pt) -- (159pt,-16pt) node[below left = 2pt and 40pt] {$N$ blocks };
\end{tikzpicture}

\caption{Physical data array $D$ for the toy construction\label{fig:toy}.}
\end{figure}
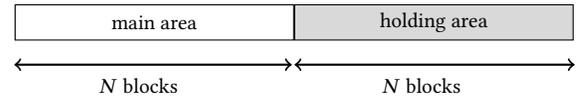

Each block of the storage area has an address $a$, and a user
interacts with the system by writing data to an address within the
main area, that is $a \in [0,N)$, and reading data back by referring
to the address. In the main area, a block is always stored at its
address, but in the holding area, blocks are appended as they are
written, irrespective of their address.

In order to track where to write the next block in the
holding area, we keep a counter $i$ of the number of write operations
performed so far.  Additionally, as the holding area is not
ordered, there needs to be a position map that associates addresses
$a\in[0,N)$
to the location in $[0,2N)$ of the freshest data associated with that address,
either at the address within the main area or a more recent write to
the holding area. The position map construction will be described later
as a write-only oblivious data structure stored with in an adjacent,
smaller WoORAM. For now, we
abstract position map as
a straightforward key-value store with operations
$\textsc{getpos}(a) \rightarrow a'$ and $\textsc{setpos}(a,a')$

With parameter $N$, counter $i$, and the WoORAM storage array $D$,
where $|D|=2N$, we can now define the primary operations of the toy
WoORAM as in Algorithm~\ref{alg:toy}. Note that the \textsc{read}
operation, which for now is trivial and which may seem irrelevant for a
write-only ORAM, is crucial to the practical performance. As we progress
to more sophisticated WoORAM schemes, both the \textsc{read} and
\textsc{write} operations will necessarily become more intricate.

\algnewcommand{\IIf}[1]{\State\algorithmicif\ #1\ \algorithmicthen}
\algnewcommand{\EndIIf}{\unskip\ \algorithmicend\ \algorithmicif}

\algrenewcommand{\algorithmiccomment}[1]{$//$\textit{#1}}
\begin{algorithm}[h]
\small
  \caption{Operations in Toy Deterministic WoORAM}
  \label{alg:toy}
  \begin{algorithmic}

  \State \Comment{Perform the $i$-th write, storing data $d$ at address $a$}
  \Function{write${}_{i}$}{$a,d$}
  \State $D[N + (i \bmod N) ] := \textsc{enc}(d)$ \Comment{Write to holding area}
  \State $\textsc{setpos}(a,N + (i \bmod N))$ \Comment{Update position map}
  \State $i := i + 1$ \Comment{increment counter}
  \\
  \State \Comment{Refresh the main area}
  \If { $i \bmod N = 0 $} 
  \For{$a \in [0,N)$}
  \State $D[a] := \textsc{enc}(\textsc{dec}(D[\textsc{getpos}(a)]))$
  \State $\textsc{setpos}(a,a)$ \Comment{Update the position map}
  \EndFor
  \EndIf
  \EndFunction
  \\

  \State \Comment{Read and return data at address $a$}
  \Function{read}{$a$}

  \Return $\textsc{dec}(D[\textsc{getpos}(a)])$  \Comment{Return freshest version of the data}
  \EndFunction

\end{algorithmic}
\end{algorithm}

\paragraph{Properties of toy construction.}
Already, our toy construction has some of the important
of the properties of our final construction. As explained below, it
provides write obliviousness, it is deterministic, and it does not
require a stash.

To see why the toy construction is write-oblivious,
first consider that each write occurs sequentially in the
holding area and has no correspondence to the actual address of the
data. Writing to the holding area does not reveal the address of the
data. Second, once the holding area has been filled completely with
the freshest data, after $N$ operations, {\em all} the main area
blocks are refreshed with data from the holding area, or re-encrypted
if no fresher data is found in the holding area.  Since all the main
area blocks are written during a refresh, it is impossible to
determine which of the addresses occurred in the holding area. In both
cases, for a write to the holding area and during a refresh, the block
writes are oblivious.

The toy construction also has a \emph{deterministic} write pattern:
the $i$-th write always touches the holding area block at index
$N + (i \bmod N)$. As compared to previous ORAM
and WoORAM schemes, in which
writing (or access) requires randomly selecting a set (or path) of
blocks to overwrite with the expectation that at least one of the
blocks has the requisite space to store the written data, our
construction does not require any random selection and operates in a
completely deterministic manner.

Further, as each write is guaranteed to succeed---we always write
sequentially to the holding area---there is no need for a stash. To
the best of our knowledge, all other WoORAM schemes require a stash to
handle failed write attempts.
In some sense, one can think of the
stash in these systems as providing state information about the
current incomplete writes, and to have a stateless system the full size
of stash would need to be transferred on every step (even if there is
nothing in it).
By contrast, our construction has constant state cost (ignoring the
position map for now), which consists
simply of the counter $i$ and the encryption key. Our construction
continues to have constant unsynchronized client state even when we
consider the de-amortized case with position map below.

\subsection{De-amortizing the toy construction}
We can advance upon the toy construction above by further generalizing
the storage procedure via de-amortization of the refresh procedure as
well as allowing the main and holding area to be of different
sizes. The key idea of de-amortization is that instead of refreshing
the main area once the holding area has been fully written, we can
perform a few writes to the main area for each write to the
holding area, so that it is fully refreshed at the same rate.

In this generalized setting, physical storage consists of a main area of
size $N$ as before, and a holding area of size $M$, where $M$ is
arbitrary, so that $|D| = N + M$.

\tikzset{
w/.style={draw,rectangle,text height=8pt,text width=150pt,align=center,fill=white},
g/.style={draw,rectangle,text height=8pt,text width=50pt,align=center,fill=gray!30},
}

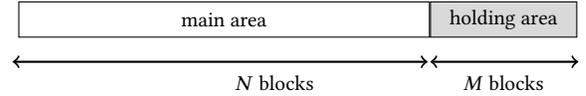
\begin{figure}[h]
\small
\centering
\noindent\begin{tikzpicture}[node distance=0]
\def\lnode#1#2{%
\node[w, label=center:{#2}, minimum width=120pt] (#1) {}}
\def\lnodex#1#2#3{%
\node[g, right = 0pt of #1, label=center:{#3},minimum width=10pt] (#2) {}}

\lnode{a}{main area};
\lnodex{a}{b}{holding area};

\draw[thick, <->] (-80pt,-16pt) -- (77pt,-16pt) node[below left = 2pt and 40pt] {$N$ blocks };
\draw[thick, <->] (78pt,-16pt) -- (134pt,-16pt) node[below left = 2pt and 10pt] {$M$ blocks };
\end{tikzpicture}
\caption{Back end data array $D$ with unequal main and holding areas.}
\end{figure}

The key to the de-amortized write procedure is that there needs to be a
commensurate number of refreshes to the main area for each write to
the holding area. After any consecutive $M$ writes to the holding area,
the entire main area (of size $N$) needs
to be refreshed, just like what would happen in the amortized toy
construction.

When $N=M$, this is simply accomplished by performing one refresh for
each write. When the sizes are unequal, we need to perform on average
$N/M$ refreshes per write to achieve the same goal. For example,
consider the case where $N=2M$, where the main area is twice as large
as the holding area, then $N/M=2$, and thus we perform two
refreshes for every write. After $M$ writes to the holding area, the
entire main area will have been refreshed.

It is also possible to have ratios where $M>N$, such as $M=2N$ where
the main area is half the size of the holding area, and in fact, this
setting and $M=N\cdot\ceil{\log(N)}$ are both critical settings for
performance. When $M>N$ this implies that we need to do {\em less}
than one refresh per write, on average. Specifically for $N/M = 1/2$,
we perform a refresh on \emph{every other write} to the holding area.

Algorithm~\ref{alg:deam-NM} has the properties of
performing on average $N/M$ refreshes per write, while the
read operation is the same as before.

\algrenewcommand{\algorithmiccomment}[1]{$//$\textit{#1}}
\begin{algorithm}[h]
\small
  \caption{De-amortized write operation with unequal size main and
  holding areas.}
  \label{alg:deam-NM}
  \begin{algorithmic}

  \State \Comment{Perform the $i$-th write, storing data $d$ at address $a$}
  \Function{write}{$a,d$}
  \State $D[N + (i \bmod M) ] := \textsc{enc}(d)$ \Comment{Write to holding area}
  \State $\textsc{setpos}(a,N + (i \bmod M))$ \Comment{Update position map}
  \\

  \State \Comment{Refresh $N/M$ main area blocks per-write}
  \State $s := \lfloor  i \cdot   N/M \rfloor  \bmod N$
  \State $e :=  \lfloor (i+1) \cdot  N/M  \rfloor  \bmod N$
  \For {$a' \in [s,e)$}{}
  \State $D[a'] := \textsc{enc}(\textsc{dec}(D[\textsc{getpos}(a')]))$
  \State $\textsc{setpos}(a',a')$
  \EndFor
  \State $i := i + 1$ \Comment{increment counter}
  \EndFunction
\end{algorithmic}
\end{algorithm}

It is straightforward to see that the unequal size,
de-amortized solution has the same key properties as the toy
construction: it is write-oblivious, deterministic,
and does not require
a stash. It is clearly deterministic because just as before, writes
and refreshes occur sequentially in the holding area and main area,
respectively, and this also assures write-obliviousness for the same
reasons discussed before. It still does not require a stash because
every write will succeed, as the refresh pattern guarantees that the
next write to the holding area will always overwrite a block that has
already had the chance to be refreshed to the main area.

The only non-obvious fact may be the \emph{correctness} of the scheme.
In particular, is it possible for some write to the holding area to
overwrite some other block which has not yet been refreshed to the main
area? The following lemma justifies that such a situation cannot happen.

\begin{lemma}
  Consider Algorithm~\ref{alg:deam-NM}.
  For any time $i$ and address $a$, there exists a time $i'$
  when address $a$ is refreshed to the main area satisfying
  $$i \le i' < i+M.$$
\end{lemma}
\begin{proof}
  Address $a$ is refreshed to main area whenever the current time $i'$
  is in the range $[s,e)$ in the for loop; namely, when
  $$\floor{\frac{i'N}{M}}\bmod M \le a < \floor{\frac{(i'+1)N}{M}} \bmod M.$$
  This happens as soon as
  $$i' \bmod M \ge \floor{\frac{Ma}{N}} \bmod M.$$
  Because this is an inequality modulo $M$ on both sides, there exists
  some $i'\in\{i,i+1,\ldots,i+M-1\}$ which satisfies it.
\end{proof}

The consequence of this lemma is that, for any time $i$, the data which
is placed in the holding area at address $N+(i\bmod M)$ will be
refreshed to the main area before time $i+M$, which is the next time
holding address $N+(i\bmod M)$ will be overwritten. Therefore no data is
overwritten before it is refreshed to the main area, and no stash is
needed.

\subsection{Incorporating the Position Map}
In this section, we consider methods for implementing a position map
for \ourworam{}, and crucially, modifying the procedure so that only a
\emph{single position map update per write} is needed.

We first describe how to modify our algorithm so that we can store the
position map recursively within successively smaller \ourworam{}s, and
then show how to further improve by using a Trie-based
\emph{write-only oblivious data structure} (WoODS) stored within an
adjacent \ourworam{} to the main, data-storing one.

\paragraph{Recursively stored position map.}
One possibility for storing the position map is to pack as many
positions as possible into a single block, and then store an adjacent,
smaller WoORAM containing these position map blocks only. Then
\emph{that} WoORAM's position map is stored in a smaller one, and so on,
until the size is a constant and can be stored in memory or refreshed on
each write. If at least
two positions can be packed into each block, the number of levels in
such a \emph{recursive position map} is $O(\log_B N)$.

If we consider each of the recursive WoORAMs using the same write
procedure as described in Algorithm~\ref{alg:deam-NM}, a problem
quickly emerges. A write requires multiple updates to the position map
due to the de-amortized procedure: one update to store the location
within the holding area of the newly written data, and some number of
updates to store the refreshed main areas. In a recursive setting, as
these position map updates must occur for every recursive level of the
position map, we can get exponential blow up. One write to the main
WoORAM requires $O((1+M/N)^R)$ writes at the smallest WoORAM, where
$R \in O(\log_B N)$ is the number of recursive levels.

In \travisworam{}, this issue is solved using additional state
information of ``metadata blocks'', each containing the actual index of
the block as well as the IV used to encrypt that block. These metadata
blocks are stored alongside the primary physical blocks for the WoORAM.
Crucially, by storing the actual \emph{index} associated with each block
in memory, it is no longer necessary to update the position map multiple
times for each write.
While something similar
would work for our system, we solve this problem more efficiently,
avoiding the need for separate storage of metadata blocks entirely.

The difference here is not asymptotic, but helps in practice by
essentially eliminating an extra metadata block read/write on every
step. It also allows us to take better advantage of the uniform block
setting, where even reading or writing a few bytes in a block requires
transferring $O(B)$ bytes of data. This technique is crucial to our
obtaining optimal $2B$ physical writes per logical write, as we show in
the next section.

\algrenewcommand{\algorithmiccomment}[1]{$//$\textit{#1}}
\begin{algorithm}[h]
\small
  \caption{\ourworam{} Operations with a Pointer Based Position Map:
    main area size $N$, holding area size $M$, data array $D$, and
    counter $i$}
  \label{alg:deam-final}
  \begin{algorithmic}

  \State \Comment{Read and return data for address $a$}
  \Function{read}{$a$}
  \State $(a_h,o,q) := \textsc{getpos}(a)$
  \State $B_m = \textsc{dec}(D[a])$
  \If{$B_m[o] = q$} ~\Return $B_m$
  \Else ~\Return $\textsc{dec}(D[a_h])$
  \EndIf
  \EndFunction
  \\

  \State \Comment{Perform the $i$-th write of data $d$ to address $a$}
  \Function{write}{$a,d$}
  \State $a_h := N + (i \bmod M) $ \Comment{Holding address}
  \State $D[a_h] := \textsc{enc}(d)$ \Comment{Write to holding area}
  \State $(o,q) := \text{\sc diff}(d,\textsc{dec}(D[a]))$ \Comment{Offset $o$ and bit diff $q$}
  \State $\textsc{setpos}(a,(a_h,o,q))$  \Comment{Update Position Map}
  \\
  \State \Comment{Refresh $N/M$ main area blocks per-write}
  \State $s :=  \lfloor i \cdot  N/M \rfloor  \bmod N$
  \State $e := \lfloor (i+1) \cdot N/M  \rfloor \bmod N$
  \For {$a_m \in [s,e)$}{}
  \State $D[a_m] := \textsc{enc}(\textsc{read}(a_m))$
  \State \Comment{No position map update needed}
  \EndFor
  \State $i := i + 1$ \Comment{increment counter}
  \EndFunction

  \\

\end{algorithmic}
\end{algorithm}

\paragraph{Position map pointers and one-bit diff technique.}
To improve the position map and remove exponential blow-up in updating
the position map, we recognize that we have a distinct advantage in
\ourworam{} construction as compared to prior schemes in that for main
area blocks, data is always located at its address. The holding area
is the only portion of the WoORAM that needs a position map. {\em The
  position map does not need to be updated for a refresh if we could
  determine the freshest block during a read.}

To see this, consider a position map that simply stores a holding-area
address. When we perform a read of address $a$, we need to look in
two locations, both in the holding area at where the position map says
$a$ is and in the main area at $a$. Given these two blocks, which is
freshest data associated with $a$?

We can perform a freshness check between two blocks using the {\em one-bit
diff} technique.
Specifically, the position map gives a mapping of logical address $a \in [0,N)$
to a tuple $(a_h,o,q)$, where $a_h \in [0, M)$ is an address to the holding
area, $o \in [0, B)$ is a bit offset within a block, and $q \in \zo$ is the bit
value of the freshest block at the offset $o$.
We define the tuple $(a_h,o,q)$ as a {\em position map pointer}.

Whenever a write occurs for logical address $a$ to holding area $a_h$,
the offset $o$ is chosen so as to \emph{invalidate} the old data at
address $a$ in the main area. Specifically, we ensure that the $o$th
bit of the new, fresh data $\dec(D[a_h])$ is different from the $o$th
bit of the old, stale data $\dec(D[a])$. (If there is no difference
between these, then the old data is not really stale and the offset
$o$ can take any valid index, say 0.)

Given the pointer $(a_h, o, q)$, a freshness check between two blocks $\dec(D[a])$
and $\dec(D[a_h])$ is performed as follows:
\begin{myquote}
Check if the $o$th bit of $\dec(D[a])$ is $q$. If so, $\dec(D[a])$ is fresh; otherwise
$\dec(D[a_h])$ is fresh.
\end{myquote}
The key observation is that {\em when a block is
refreshed to the main area, there is no need to update the position
map with a new pointer}, since the read operation always
starts by checking if the block in the main area is fresh. If the main
area block is fresh, then there is no need to even look up the holding
area position (which may have been rewritten with some newer block for a
different logical address).
See Algorithm~\ref{alg:deam-final} for details of how this is accomplished.

\paragraph{Trie WoODS for Position Map}
\label{sec:trie}
A more efficient solution for storing the position map, as compared to the
recursively stored position map, is to use an oblivious data structure (ODS)~in
the form of a Trie.
Recall that Trie edges are labeled, and looking up a node with a keyword $w_1
w_2 \cdots w_\ell$ is performed by starting with a root node and following the
edge labeled with $w_1$, and then with $w_2$, and all the way through the edge
labeled with $w_\ell$ one by one, finally reaching the target node.

As with previous tree-based ODS schemes~\cite{CCS:WNLCSS14,SP:RocAviCho16}, our
ODS scheme avoids recursive position map lookups by employing a {\em
pointer-based technique}. That is, the pointer to a child node in a
Trie node {\em directly points to a physical location} instead of a
logical location, and therefore it is no longer necessary to translate
a logical address to a physical address within the Trie itself.

Applying an ODS in a write-only setting (a WoODS or write-only
oblivious data structure) is similar to an idea proposed by
Chakraborti~et. al~\cite{PETS:ChaCheSio17}.
A major difference in our construction
is that we do not store the data structure within the primary
WoORAM.
We also allow changing the branching factor of the Trie independently of the
block size, so we can tune the secondary WoORAM and flexibly control the
number of physical block writes for every logical write, including
position map information stored within the Trie.

\def\wwrite{\textsc{write-to-holding}}
\begin{figure}[t]
\centering
\noindent\includegraphics[width=0.8\linewidth]{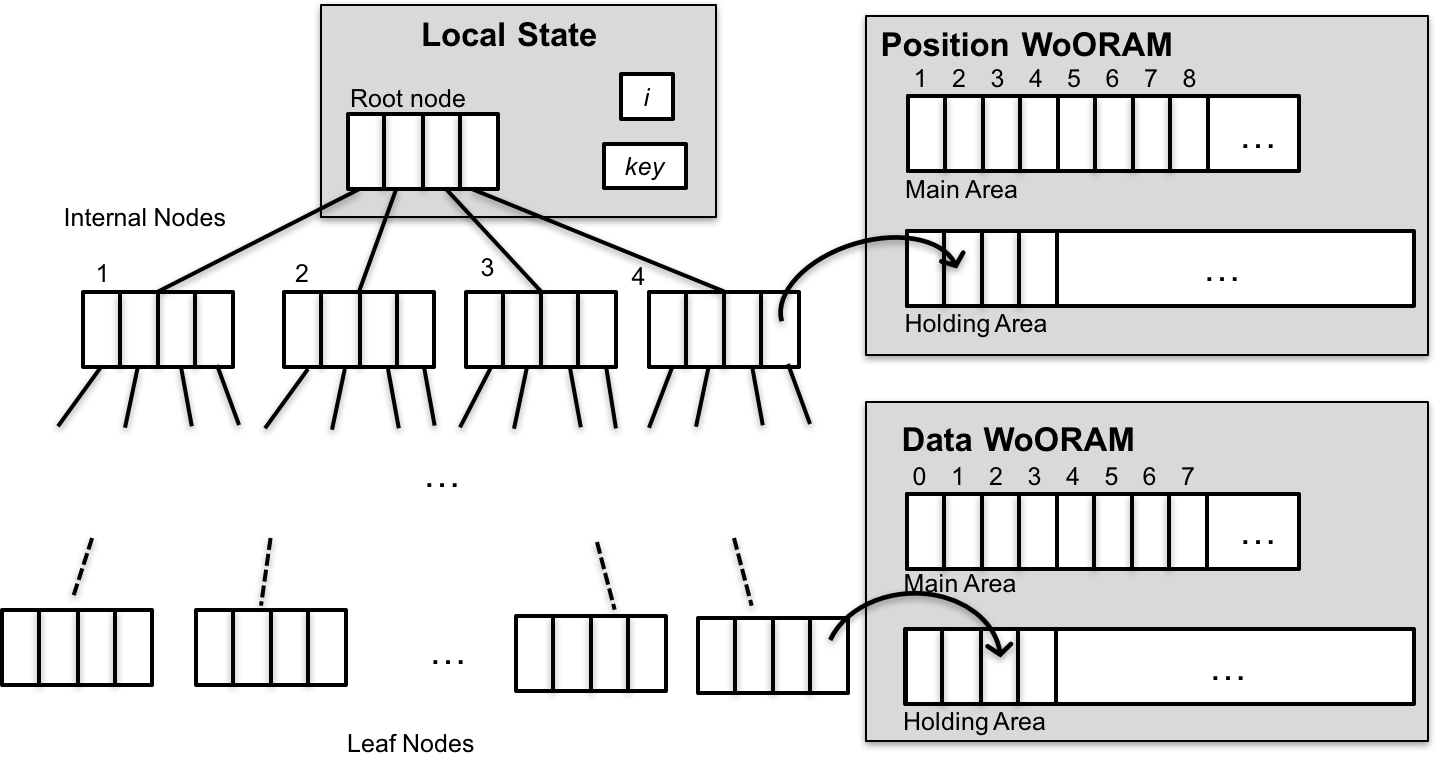}
\caption{\ourworam{} Diagram with Trie Position Map}
\label{fig:trie}
\end{figure}

\ifpreprint
\begin{algorithm}[p]
\else
\begin{algorithm}[h]
\fi
\small
  \caption{Trie WoODS with $N_p$ nodes and branching factor $b$.
   \textsc{read} and \textsc{write} calls are the routines in
   Algorithm~\ref{alg:deam-final} (with modified subroutines as specified)
   applied to position WoORAM instantiated with $N_p$ main blocks and $M_p$
   holding blocks.}
  \label{alg:deam-trie}
  \begin{algorithmic}

  \State \Comment{$a \in [0, N_p+N)$ is a position WoORAM or data WoORAM
  address}
  \Function{path-indices}{$a$}
  \If{$a = 0$} \Return $[]$ \Comment{Base case: empty path to root node}
  \Else\ \Return $[\textsc{path-indices}(\lfloor (a-1)/b \rfloor),
    (a-1)\bmod b]$
  \EndIf
  \EndFunction
  \\

  \State \Comment{Retrieve Trie nodes along a path}
  \Function{path-nodes}{$a_0, a_1, \ldots, a_{\ell-1}$}
  \State $B_0 := \textrm{root node}$ \Comment{Root node is kept in
  local state}
  \State $a := 0$
  \For{$i=0, \ldots, \ell-1$}
  \State $ptr := B_i[a_i]$ \Comment{The pointer to the $a_i$th child, i.e., $ptr = (a_p, o, q)$}
  \State $a := (a_i+1) + b \cdot a$
  \State $B_{i+1} := \textsc{read}(a)$ in Alg.~\ref{alg:deam-final} with its subroutine changed as:
  \State \phantom{...} $\rhd$ $\textsc{getpos}(a)$ returns $ptr$
  \EndFor
  \State \Return $(B_0, \ldots, B_{\ell})$
  \EndFunction
  \\

  \State \Comment{$a$ is an address; $src$ is either \texttt{DATA} or
  \texttt{TRIE}}
  \Function{getpos-trie}{$a, src$}
  \If{$src = {\tt DATA}$} $a_0,a_1,\ldots,a_\ell := \textsc{path-indices}(N_p+a)$
  \Else\  $a_0,a_1,\ldots,a_\ell := \textsc{path-indices}(a)$
  \EndIf
  \State $(B_0, \ldots, B_{\ell}) := \textsc{path-nodes}(a_0, a_1, \cdots a_{\ell-1})$
  \State \Return $B_{\ell}[a_\ell]$
  \EndFunction
  \\

  \State \Comment{$a$ is a data WoORAM index; $ptr$ is a pointer}
  \Function{setpos}{$a,ptr$}
  \State $a_0,a_1,\ldots,a_\ell := \textsc{path-indices}(N_p+a)$
  \State $(B_0, \ldots, B_\ell) := \textsc{path-nodes}(a_0, a_1 \cdots a_{\ell-1})$
  \State $B_\ell[a_\ell] := ptr$ \Comment{Change the leaf first}
  \For{$j = \ell, \ldots, 1$}{} \Comment{from leaf to root}
    \State Call \textsc{write}($a_{j-1}, B_j$) in Alg.~\ref{alg:deam-final} with its subroutines changed as:
    \State \phantom{...} $\rhd$ \textsc{setpos}($a_{j-1}, ptr'$) assigns $ptr'$ to $B_{j-1}[a_{j-1}]$
    \State \phantom{...} $\rhd$ \textsc{getpos}($a$) returns
    \textsc{getpos-trie}($a,\texttt{TRIE}$)
  \EndFor
  \IIf{$\ell \ne \ceil{\log_b(N_p)}$}
    write a dummy Trie node
  \EndIIf
  \State \textrm{Root node} $:= B_0$
  \EndFunction

\end{algorithmic}
\end{algorithm}

As the WoODS Trie is stored in an adjacent \ourworam{} construction,
we differentiate between the two WoORAMs by referring to the {\em data
  WoORAM} as the WoORAM storing data blocks and the {\em position
  WoORAM} as the WoORAM storing the nodes of the Trie. The Trie itself
acts as the position map, and will map addresses in the data WoORAMs
main area to position map pointers referencing the data WoORAMs
holding area. The main idea is that given an address $a$, one can walk
the Trie to find a leaf node storing $a$'s
position map pointer. The position WoORAM will be strictly smaller than the data
WoORAM, but will be implemented using the same \ourworam{} framework (i.e.,
using the notions of main area, the holding area, and the counter).

Details of
the procedure for the position WoORAM is outlined in
Algorithm~\ref{alg:deam-trie}. Observe that the functions for the
position WoORAM call the \textsc{READ} and \textsc{WRITE} functions
from Algorithm~\ref{alg:deam-final}, but with modified versions of
the subroutines for accessing and updating the position map, as the Trie
is its own position map.

As noted, the Trie is stored in an adjacent WoORAM that has a main
area and a holding area. The key difference is that the Trie nodes are
addressed with the position WoORAM's main area using heap
indexing. For example, with a branching factor of $b=4$, the root node of the
Trie has address 0, its children are at address 1, 2, 3, and 4, and their
children are at addresses (5, 6, 7, 8), (9, 10, 11, 12), and so on. By using heap
indexing, the structure of the Trie reveals the position of its nodes, becoming
its own position map. In particular, this indexing avoids the need to
store edge labels explicitly; they can instead be stored implicitly
according to the heap indexing formulas.

It is still possible for a node of the Trie to have been
recently updated and thus the freshest node information to be resident
in the holding area of the position WoORAM.  As such, each internal
node of the Trie stores $b$ position map pointers to the position
WoORAM's holding area, one for each of its child nodes. A leaf node
in the Trie then stores $b$ position map pointers to the data WoORAM's
holding area. The root node of the Trie can be stored as part of the local state,
since it is constantly rewritten and read on every operation.
A visual of the Trie is provided in Figure~\ref{fig:trie}.

Reading from the Trie to retrieve a position map pointer for the data
WoORAM is a straightforward process. One only needs to traverse from
the root node to a leaf, following a path dictated by the address $a$
called via $\textsc{getpos-trie}(a,\texttt{DATA})$.
On each step down the tree, the
current Trie node stores the position map pointer of the child node;
the corresponding sequence of nodes are retrieved via the
\textsc{path-nodes} helper function.
The position map pointers for the data WoORAM can be found at the
correct index in the leaf node along the fetched path.

Updating a pointer in the Trie (by calling $\textsc{setpos}(a,ptr)$)
is a bit more involved.
An update of the position map for the data WoORAM
requires updating a leaf node in the Trie within the position
WoORAM. Writing that leaf node will change its pointer, which requires
updating the parent node, whose pointer will then also change, and so on
up to the root of the Trie.
That is, \emph{each write to the
main WoORAM requires rewriting an entire path of Trie nodes within
the position WoORAM}.

Recall that in \ourworam{}, each write operation not only writes one
block to the holding area, but also performs some \emph{refreshes} in
the main area. The challenge is, for each refresh, determining where in
the holding area fresher data might be. For the data WoORAM, this is
achieved simply by performing a lookup in the position map. But for the
position WoORAM, there is no position map! Instead, we use the Trie
itself to look up the pointer for fresher data in a position WoORAM
refresh operation, by calling $\textsc{getpos-trie}(a,\texttt{TRIE})$.
This is possible again because of the heap indexing;
from the index of the Trie node that is being refreshed, we can
determine all the indices of the nodes along the path to that one, and
then perform lookups for the nodes in that path to find the position
WoORAM holding area location of the node being refreshed.

\paragraph{Trie WoODS parameters and analysis}
We start by calculating $N_p$, which is the number of Trie nodes as well
as the size of the position WoORAM main area.
This needs to be large enough so that there is room for $N$ pointers in
the leaf nodes, where $N$ is the number of logical addresses in the data
WoORAM. With branching factor $b$, the number of Trie nodes is given by
\begin{equation}
  N_p = \left \lfloor \frac{N-2}{b-1} \right \rfloor
    \in O(\tfrac{N}{b-1}).
  \label{eqn:Np}
\end{equation}

To derive \eqref{eqn:Np} above, consider that each Trie node holds $b$
pointers, either to children in the Trie or to addresses in the data
WoORAM. We do not count the root node in $N_p$ because it changes on
each write in is stored in the $O(1)$ client local memory. The total
number of pointers or addresses stored is therefore $N_p+N$.
This leads to the inequality
$(N_p+1)b \ge N_p + N,$
which implies
$N_p = \ceil{\frac{N-b}{b-1}}.$
The form of \eqref{eqn:Np} is a simple rewriting of this floor into a
ceiling based on the fact that, for any two integers $x,y$,
$\ceil{x/y} = \floor{(x+y-1)/y}.$

The height
of the Trie is then equal to the height of a $b$-ary tree with $N_p+1$
nodes, which is $O(\log_b N)$. This is the number of
Trie nodes that need to be written to the holding area of the position
WoORAM on each update (including a potential dummy node).

Each write to the data WoORAM requires rewriting
$h\in O(\log_b N)$ nodes in the Trie (for a single path).
Each of those writes to the holding area of the position
WoORAM needs, on average, $N_p/M_p$ number of refreshes to the main
area, where $N_p$ is the size of the position WoORAM's main area and
$M_p$ is the size of the position WoORAM's holding area.

Looking up a position in the Trie requires reading $O(\log_b N)$
blocks in the position WoORAM, each of which results in up to two
physical reads due to having to check for fresher data.
A refresh operation in the position WoORAM also requires
a read of the Trie to determine if fresher data for that node exists
in the holding area. If the sizes of the main area $N_p$ and holding
area $M_p$ for the position WoORAM are not set appropriately, this could
lead to $O(\log^2 N)$ reads to perform an update. However, consider
that we can control the ratio $N_p/M_p$. If we set
$M_p \gg N_p \log_b N$ then we need to perform only $O(1)$
refreshes per position map update, thus requiring
$O(\log N_p)$ reads per update.

If $N$ is a power of $b$, then the number of leaf nodes in the Trie
$N/b$ is also a power of $b$, and the Trie is a complete $b$-ary tree of
height $\log_b (N/b)$. If $N$ is not a power of $b$, then the last level
of the Trie is incomplete, and leaf node heights may differ by one. In
order to preserve write obliviousness, in cases of rewriting a path with
smaller height, we add one additional dummy node write.

Finally, observe that the branching factor $b$ can play a role in
the performance. With $b=2$, the \emph{size} of a Trie path is
minimized, but the height and number of nodes $N_p$ are maximal.
Increasing $b$ will reduce the height of the Trie and the number of Trie
nodes, while increasing the total size of a single path.
As we will show next, adjusting the {\em packing} of position map
by setting the branching factor $b$ can be done carefully to achieve
write communication cost of exactly $2B$ in a fully sequential write
pattern.

\subsection{Fully Sequential Physical Write Pattern}
\label{sec:packing}
In this section, we describe how to achieve fully sequential writing
of physical storage and how to minimize the communication cost. This
requires interleaving the various storage elements of \ourworam{} such
that all the writing, regardless of which part of the construction is
being written, occurs sequentially.

To understand the challenge at hand, first consider a simple
implementation which aligns the data WoORAM (main and
holding areas) adjacent to the position WoORAM (main and holding area)
forming a single storage data array broken into size-$B$ blocks. A
write to the data WoORAM will result in a write to the holding area
plus $M/N$ average writes to the main area. The position map is also
updated, requiring $O(\log_b N)$ nodes in the Trie to be written to
the holding area and $O(1)$ refreshes of the position WoORAM's main
area, provided $M_p \in \Omega(N_p\log_b N)$.
While all these writes occur sequentially within their respective
data/position WoORAM main/holding areas, do \emph{not} occur
sequentially on the underlying storage device as each of the various
WoORAM areas are separated. Furthermore, the writes to the position map
are wasteful in that they may update only a few nodes, constituting just
a small fraction of the block, but in the uniform access model this in
fact requires updating the entire block.

We can improve on this storage layout and achieve a minimum
in write performance requiring exactly $2$ blocks to be written to
physical storage for each block write, where one block is the new data,
a half block worth of main area refresh, and a half block worth of
position map updates. Further, we can interleave the various WoORAM
portions such that those $2$ blocks are written sequentially on the physical device.

\paragraph{Data WoORAM Block Interleaving}

Every logical block write to the data WoORAM results in exactly one
block write to the holding area of data WoORAM.
Recall that there are two parameters for setting up \ourworam{}: $N$,
the size of the main area, and $M$, the size of the holding area. These
two values need not be the same, and in fact, to achieve sequential
writing, we will set $M = 2N$.
In this case, on average $\frac{M}{N}=\frac{1}{2}$ block is refreshed to
the main area for each logical write.

With adjacent main and holding areas, this could be achieved by
performing one full block refresh on \emph{every other} logical write.
To make the writing sequential, we will instead refresh \emph{half} of a
block on every logical write, resulting in the following storage layout:
\begin{align*}
h_0,\cfrac{m^0_0}{\square},h_1,\cfrac{m^1_0}{\square},h_2,\cfrac{m^0_1}{\square},h_3,\cfrac{m^1_1}{\square}, \ldots, h_{M-2},\cfrac{m^0_{N-1}}{\square}, h_{M-1},\cfrac{m^1_{N-1}}{\square}
\end{align*}
where $m^0_j$ is the first half of the block $m_j$, $m^1_j$ is the
back half of $m_j$, and $\square$ represents empty space.
(This empty space will be used to store nodes for the position WoORAM,
as we will show next.)

There is a slight complication to reading now, as a single main area
block is actually divided between two physical memory locations,
resulting in an additional (constant) overhead for reading operations.
The benefit is that the writing is fully sequential now: each logical
write requires writing sequentially the data being updated (to the
holding area), and the next \emph{half} block of data being refreshed
(to the main area), plus another half-block containing position map
information as we will detail next. Also observe that, under this
configuration with $M=2N$, the total physical memory requirement will be
$4N$ blocks.

\begin{figure}[t]
\centering
\noindent\includegraphics[width=0.8\linewidth]{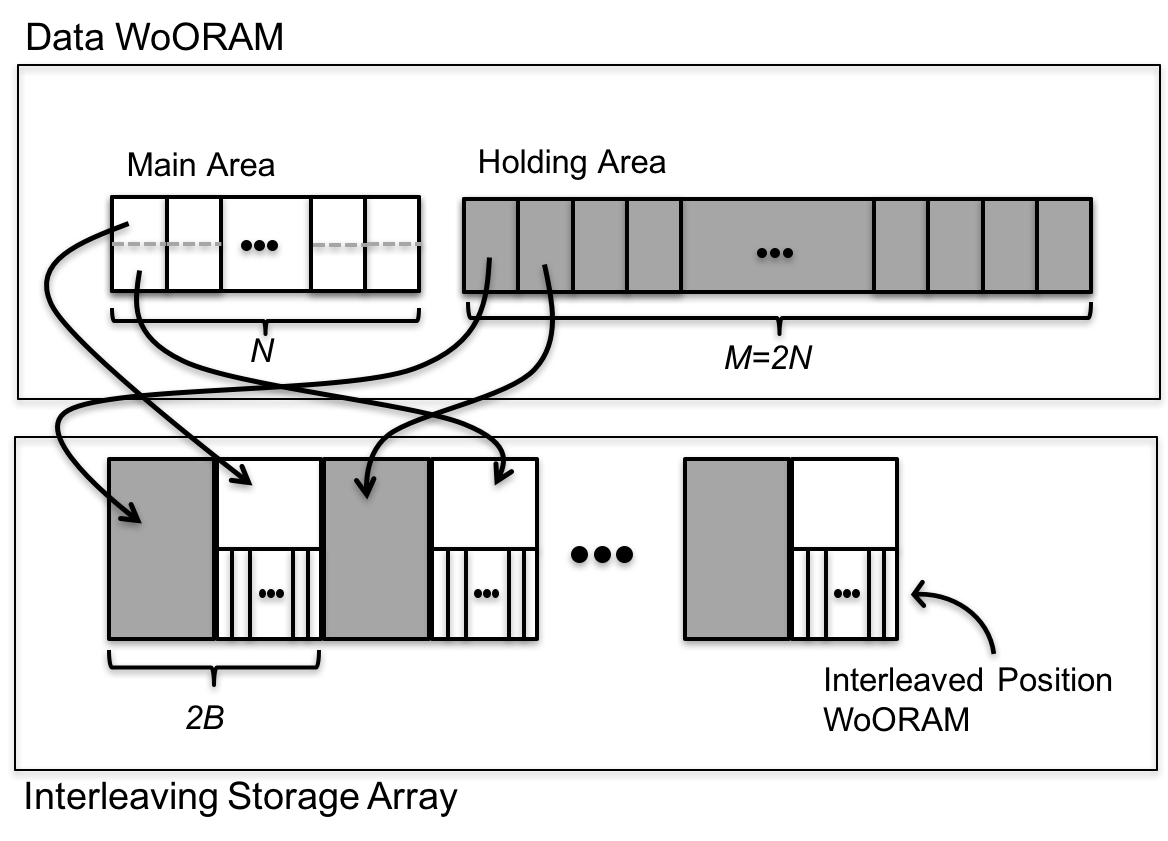}
\caption{Interleaving of WoORAMs into a Storage Array with $2B$ size sequential writes.}
\label{fig:interleave}
\end{figure}

\paragraph{Position WoORAM Block Interleaving}
As suggested above, the remaining half-block of data $\square$ in the
above construction will be used to store position map information.
A diagram storage achieving 2 sequential physical block writes per
logical write is shown in Figure~\ref{fig:interleave}.

Specifically, these $\frac{B}{2}$ bits will store the Trie nodes
written to position WoORAM holding and main areas during a single
logical write operation.
This is (potentially) possible because the Trie nodes in the position
WoORAM are much smaller than the blocks in the data WoORAM.
Fully sequential writing will be achieved if and only if all of the
Trie nodes written during a single step
can always fit into $B/2$ bits.

There are many settings of parameters $M$, $b$, and $M_p$ that may make
sequential writing possible, depending on logical and physical memory
requirements and the physical block size $B$. We will choose some
parameters here and demonstrate that they would work for any conceivable
value of $N$.

For this purpose, set the branching factor $b=2$, and then recall from
\eqref{eqn:Np} that the number of Trie nodes and
the Trie height will be
$N_p = N-2$ and $h = \ceil{\lg (N-2)}$, respectively.

Next, set the number of nodes in the position WoORAM holding area to be
\begin{equation}
M_p = N_p \cdot h.
\label{eqn:Mp}
\end{equation}
This ensures that only (at most) one Trie node needs to be refreshed to
the position WoORAM main area when writing an entire path of $h$ Trie
nodes during a single logical write operation. (The number of Trie nodes
written to the holding area for each operation is always $h$.)
Based on these formulae, we need to have enough space in the $B/2$ bits
of a half block to fit $h+1$ Trie nodes.

What remains is to estimate the size of each Trie node. Each node stores
$b=2$ \ourworam{} pointers, each of which contains $\ceil{\lg M}$ bits
for the holding area position, $\ceil{\lg B}$ bits for the block offset,
and 1 bit for the bit diff value. The condition that $h+1$ Trie nodes
fit into $B/2$ bits then becomes
\begin{equation}
\left(h+1\right) \cdot \left(\ceil{\lg M} + \ceil{\lg B} + 1\right)
\le \frac{B}{2}
\label{eqn:fitblock}
\end{equation}

Combining this inequality with all of the previous settings for $b$,
$M$, and $M_p$,
and assuming a block size of 4096 bytes (so $B=4096\cdot 8 = 2^{15}$)
as is the default in modern Linux kernels, we have
$$
\left(\ceil{\lg\left(N-2\right)}+1\right)
\cdot
\left(\ceil{\lg N} + 17\right)
\le 2^{14}.
$$
That inequality is satisfied for values of $N$ up to
$6.6\times 10^{35}$, which is \emph{much} more than any conceivable storage 
size. Further tuning of the $b$ and $M$ parameters could be dome to
achieve an even tighter packing and/or better read performance while
maintaining 2 physical block writes per logical write.

\subsection{Encryption Modes}
The deterministic and sequential access pattern fits nicely with encryption of
each block using \emph{counter mode}. In particular, We encrypt each
\ourworam{} block using AES encryption based on the number $i \|
0^{64}$ as a counter. Recall that the client maintains the global counter $i$
(64-bit long). Assuming the block size $B$ is reasonable (shorter than $2^{64}
\cdot 16$ bytes), there will be no collision of IVs, and the security of
encryption is guaranteed.  We stress that we do not need space for storing IVs
due to this optimization which cannot be applied to previous schemes. For
example, the randomized procedures in schemes like \travisworam{}, IVs must be
stored separately on the server, adding to the communication cost overhead.

However, the physical blocks that store the position map Trie are
encrypted with AES in CBC mode. When we pack
multiple Trie nodes together, such as during interleaving or packing as
described previously, we can encrypt a group of Trie nodes together in one shot
using a single IV. Since Trie nodes are much smaller than $B$, we can place
that IV for that group of nodes at the beginning of the block itself,
thus avoiding an extra memory access on read or write.

We note that after packing the Trie nodes into blocks, the number of blocks in
the main \ourworam{} is significantly larger than that in the
Position-\ourworam{}, so that most of the data is encrypted using counter
mode.

\section{Analysis of \ourworam{}}

We formally state the security (obliviousness), and
communication complexity of \ourworam{}. Fortunately, the simplicity of
the construction makes the proofs relatively straightforward in all
cases.

\paragraph{Security proof}
First we state the security in terms of the definitions in
Section~\ref{sec:woramdef}.

\begin{theorem}
  \ourworam{} provides write-only obliviousness.
\end{theorem}
\begin{proof}
  Let $\vec{x}$ and $\vec{y}$ denote two sequences of operations in
  \ourworam{} that contain the same number of write operations.

  The sequence of \emph{locations} of physical writes is deterministic and does not
  in any way depend on the actual locations being written, and the
  \emph{contents} of physical writes are encrypted using an IND-CPA
  symmetric cipher. Therefore ii
  it holds that
$$\Filter(\Access^\Pi(\vec{x})) ~~
    \approx_c ~~ \Filter(\Access^\Pi(\vec{y})), $$
  because the locations in these two access patterns are identical, and
  the contents in the access patterns are indistinguishable from random.
\end{proof}

\paragraph{Communication complexity}
For the complexity analysis, assume that:
\begin{itemize}
  \item the size ratio $M/N$ is a constant,
  \item the branching factor $b$ is a constant,
  \item the block size $B$ is large enough to
    contain a single path of trie nodes, and
  \item the position map holding area is at least $O(\log N)$ times
    larger than the position map's main area.
\end{itemize}

Asymptotically this means
that $B \in \Omega(\log^2 N)$.
From a practical standpoint, even in the extreme case of storing a
yottabyte of data ($2^{80}$ bytes), with holding area size $M=N$,
branching factor $b=2$, and 4KB
blocks (i.e., $B=4096$), an entire path of trie nodes is still well
below the block size at 1496 bytes.

\begin{theorem}
Under the assumptions above,
the number of physical block writes per logical block write in
\ourworam{} is $O(1)$. Furthermore, the number of physical block
reads per logical read or write operation is $O(\log N)$.
\end{theorem}
\begin{proof}
Let $h\in O(\log_b N)$ for the height of the trie that stores the
position map.  A single read to \ourworam{} requires at most two block
reads and one position map lookup, which requires fetching all $h$
nodes in the Trie path to that position. Fetching a Trie node in the
position WoORAM requires accessing the parent node as well, requiring
at most $2h$ nodes need to be fetched.  In the worst case every node
might be packed in a different block, so this is $O(1 + h)$ physical
block reads per logical read, which is $O(\log_b N)$.

A single write to \ourworam{} requires at most $1 + \lceil M/N \rceil$
block writes to holding and main areas and one update to the position
map Trie. Each main area refresh requires an additional block read and
position map lookup.  Because $M/N$ is a constant, this is $O(1)$
block writes, $O(\log_b N)$ reads, and one trie update.

Updating a single node in the Trie involves first fetching the path to
that node in $O(\log_b N)$ physical reads, then writing each node on
that path, updating the pointers in the parent nodes from leaf back up
to root. This requires $h$ writes to the position map WoORAM, which from
the assumptions will require $h$ writes to the holding area plus $O(1)$
refreshes in position map WoORAM's main area. These $O(1)$ refreshes
each require looking up $O(\log_b N)$ nodes in the position map
WoORAM, for an additional reading cost of up to $O(\log N)$ physical
blocks.

All together we get $O(\log_b N)$ physical reads per logical write,
and $O(1)$ physical writes per logical write.
\end{proof}

\section{Implementation}

We have implemented our \ourworam{} system, using the Trie-based
position map, in an open source C++ library available at
\url{https://github.com/detworam/detworam}. As we will show in this
section, comparison benchmarks validate our theoretical results on the
efficiency of \ourworam{}, showing it to be many times faster than the
previous scheme of \travisworam{}, and only a few times slower than a
non-oblivious baseline.

\paragraph{Our library}
The library relies on BUSE (Block device in USErspace,
\url{https://github.com/acozzette/BUSE}) to allow mounting a normal
filesystem as with any other device. We also use the
mbed TLS library (\url{https://tls.mbed.org/}) for encryption utilities.
We also made extensive use of C++ templates in our implementation, which
allows for considerable flexibility in choosing the parameters for the
\ourworam{} and automatically tuning the performance at compile-time.
For example, based on the size and number of backend storage blocks, the
exact byte sizes needed to store pointers, relative proportion of data
WoORAM to position map WoORAM, trie height, and relative main/holding
area sizes will all be seamlessly chosen at compile time.

\begin{table}[tb]
\ifpreprint
\centering\noindent
\else
\resizebox{\columnwidth}{!}{%
\fi
\begin{tabular}{r|cc|cc|}
  & \multicolumn{2}{c|}{Sequential write}
  & \multicolumn{2}{c|}{Sequential read}
\\
  & MB/sec & overhead
  & MB/sec & overhead
\\
\hline
\textbf{dm-crypt baseline} &&&& \\
SSD & 505.4 & ---
  & 615.8 & ---
  \\
HDD & 111.6 & ---
  & 126.1 & ---
  \\
\hline
\textbf{\travisworam{}} &&&& \\
SSD & 2.6 & 192x & 40.5 & 15.2x \\
\hline
\textbf{\ourworam{}} &&&& \\
SSD, $M=3N$, $b=64$
  & \textbf{49.7} & 10.2x & \textbf{260.0} & \textbf{2.37x}
  \\
SSD, $M=N$, $b=64$
  & 34.0 & 14.8x & 244.1 & 2.52x
  \\
HDD, $M=3N$, $b=64$
  & 29.0 & \textbf{3.84x} & 26.1 & 4.83x
  \\
HDD, $M=N$, $b=64$
  & 25.0 & 4.46x & 24.1 & 5.24x
  \\
\hline
\end{tabular}
\ifpreprint
\else
}
\fi

\ifpreprint
\small\centering
\else
\raggedright
\fi
Logical disk size 40GB and block size 4KB in all cases.\\
Overhead is relative to the dm-crypt baseline for that drive
type.\\
Highlighted values indicate the best WoORAM per column.
\caption{bonnie++ benchmarking of sequential accesses%
\label{tab:bonnie}}
\end{table}

The implementation is exactly as described in the previous section, with
a default Trie branching factor of $b=64$ unless otherwise noted.
The only exception is that we did \emph{not} implement the full
interleaving, but rather the packing solution within the position map
WoORAM to pack trie nodes into single blocks. Two blocks at a time
(from the
position map holding and position map main areas) are held in memory
while they are being filled sequentially, and then are written back to
disk once filled. In total, the result is that rather than having a
fully sequential access pattern as we would with full interleaving, we
see 4 sequential write patterns in sub-regions of memory.

\begin{table}[tb]
\ifpreprint
\centering\noindent
\else
\resizebox{\columnwidth}{!}{%
\fi
\begin{tabular}{r|cc|cc|}
  & \multicolumn{2}{c|}{Solid state SSD}
  & \multicolumn{2}{c|}{Spinning platters HDD}
\\
  & MB/sec & overhead
  & MB/sec & overhead
\\
\hline
dm-crypt baseline
  & 154 & ---
  & 16.4 & ---
  \\
\travisworam{}
  & 8.49 & 18x
  & 0.051 & 325x
  \\
\ourworam{}
  & 34.4 & 4.5x
  & 10.2 & 1.6x
  \\
\hline
\end{tabular}
\ifpreprint
\else
}
\fi

\ifpreprint
\small\centering
\else
\raggedright
\fi
Logical disk size 40GB and block size 4KB in all cases.\\
Overhead is relative to the dm-crypt baseline for that drive
type.\\
\ourworam{} used $M/N=3$ and $b=64$ for all cases.\\
Highlighted values indicate the best WoORAM per column.
\caption{fio benchmarking of random reads and writes%
\label{tab:fio}}
\end{table}

\paragraph{Comparisons}
We carefully re-implemented the \travisworam{} (only the
WoORAM part, not the hidden volume part), using the same BUSE/mbedTLS
library setup. As in their original paper and implementation, our
\travisworam{} implementation uses $k=3$ random physical writes per
logical write, and makes use of a recursive position map.
The original implementation was as a kernel module for a
device mapper, but unfortunately due to Linux kernel changes this module
is incompatible with recent Linux kernels. In fact, this incompatibility
was part of our motivation to use only standards-compliant userspace C++
code for our \ourworam{} implementation.

For a baseline comparison, we wanted to use the best existing solution
with the same general setup as ours. Our baseline uses the linux kernel
module \texttt{dm-crypt}, which provides an encrypted block device with
no obliviousness, connected to simple ``passthrough'' device that comes
with the BUSE distribution.
There is no obliviousness in this option; it simply encrypts/decrypts and stores the
resulting ciphertext in the same location on disk. This provides a fair
baseline to our work, and should help to eliminate any bottlenecks or
artifacts of the BUSE layer in order to have a clear comparison with our
new \ourworam{} protocol.

\paragraph{Measurement using bonnie++}
Table~\ref{tab:bonnie} shows the results of running the popular
\texttt{bonnie++} disk benchmarking tool our the plain encryption
as well as different WoORAM settings. All tests were performed with a
40GB logical filesystem within a 200GB partition, using the
\texttt{btrfs} filesystem.

We tested using 200GB partitions on a 1TB HDD (HGST Travelstar 7200RPM)
and on a 256GB SSD (Samsung 850 Pro). We note that both drives are
standard commodity disks available for around \$100 USD.
As expected, the SSD drive is considerably faster for both reading and
writing.

Recall that one novel feature of \ourworam{} is that it can flexibly
adapt to different storage ratios between logical and physical storage.
We tested both with $M=N$, similar to the \travisworam{}, and with more
physical space of $M=3N$, and observed a slight (but statistically meaningful)
performance improvement from having more physical disk space (and
therefore larger holding area in the \ourworam{}). We also tested with
different branching factors ranging from $b=2$ to $b=512$, but did not
notice any significant timing differences overall, indicating that the
position map plays a smaller role in the overall performance.

Overall we can see that the \ourworam{} suffers only a 3x-10x slowdown
compared to the baseline, whereas the \travisworam{} is almost 200x
slower in the case of writing and 15x slower for reading compared to the
same baseline. The results for \travisworam{} are consistent with the
results reported in their original paper \cite{CCS:BMNO14}.

In fact, our \ourworam{} running on an SSD is in most
cases \emph{faster} than the baseline running on a spinning disk HDD,
providing good evidence that our system is fast enough for practical
use. We believe this is largely explained by the \emph{sequential write
pattern} of \ourworam{}, which also makes read operations
\emph{partially} sequential. For large sequential workloads, the data
locality appears to have a very significant effect on performance.

\paragraph{Measurement using fio}
As has been noted in previous WoORAM works \cite{CCS:BMNO14}, performing
sequential \emph{logical} operations can put WoORAMs in an especially bad
light, as the baseline non-oblivious storage will translate the
sequential read/write operations to physically sequential addresses,
thereby gaining significantly over WoORAMs that need to obscure the
logical address of each operation.

Interestingly, our \ourworam{} is a somewhat ``in-between'' case here,
as the write pattern is completely sequential, and the read pattern is
\emph{partially} sequential: the main area of storage corresponds
exactly to physical addresses, but the holding area and position map do
not. We used a second disk performance measurement tool \texttt{fio}
(\url{https://github.com/axboe/fio}) in order to perform \emph{random}
reads and writes, as opposed to the sequential read/write pattern of the
\texttt{bonnie++} benchmarks. The results are shown in
Table~\ref{tab:fio}, which shows the throughput for random reads and
writes of 4KB-4MB sized blocks in direct access to the device without
any filesystem mounted.

As expected, the performance degradation for HDD compared to SSD in all
cases was significant for the random reads and writes. As with the
bonnie++ benchmarks, but more dramatically here, our \ourworam{}
running on an SSD outperformed the baseline running on the HDD. Even
more surprisingly, on the HDD our \ourworam{} was only 1.6x slower than
the baseline.  This can be explained in part by the fact that our scheme
actually \emph{turns random writes into sequential writes}, so although
it performs more writes than the baseline, they will be more compact in
this experiment.

\section{Insecurity of DataLair}
A recent paper \cite{PETS:ChaCheSio17} has also proposed
to improve the performance of \travisworam{}. While this paper contains
some new and promising ideas, and in particular
proposed the use of a B-tree ODS similar to our Trie ODS for the
position map, unfortunately it violates the notion of write-only
obliviousness.

Intuitively, the DataLair scheme identifies that a bottleneck in
\travisworam{} is in identifying free blocks from the random blocks
chosen, and propose to modify the random block choosing scheme in order
to find free blocks more efficiently with fewer dummy writes.
Unfortunately, this improvement leaks a small amount of
information about which blocks are free or not, and thereby allows an
adversary to distinguish between whether recent writes have been to the
same address, or to different addresses. We formalize this notion and
prove the insecurity of these schemes below.

We note that, since the submission of this work, the authors of
\cite{PETS:ChaCheSio17} have acknowledged the vulnerability here and
proposed a fix as a preprint \cite{datalair-arxiv}.

\label{sec:attack}

\paragraph{Overview of scheme}
\label{apx:attack}
Let $N$ be the number of logical blocks. DataLair sets $2N$ to the number of
physical blocks so that the number of free physical blocks is always $N$.  
In DataLair~\cite[Section IV]{PETS:ChaCheSio17}, every ORAMWrite considers two disjoint sets of $k$ items: 
\BI
\item Free set $S_0$: A set of $k$ blocks chosen randomly among the $N$ free physical blocks.
\item Random set $S_1$: A set of $k$ blocks chosen randomly among the entire $2N$ physical blocks. 
\EI

To make sure that $S_0$ and $S_1$ are disjoint, some elements may be removed
and addional steps of sampling may be done.
Based on the two sets, the ORAM writes a data block as follows:

\begin{framed}
\small
ORAMWrite($d$): // $d$ is a data block
\BEN
\item Insert $d$ in $\stash$
\item Create two sets $S_0$ and $S_1$ as described above.
\item Choose $k$ blocks $U = \{u_1, \ldots, u_k\}$ as follows: \\
      For $i = 1$ to $k$: \\
      \phantom{...} $b_i \from \zo$, and fetch (and remove) $u_i$ from $S_{b_i}$. \\
      \phantom{...} If $b_i=0$ and $\stash$ is not empty: \\
      \phantom{......} Take out a data item from $\stash$ and write it in $u_i$. \\ 
      \phantom{...} Otherwise, reencrypt $u_i$.
\EEN
\end{framed}

We assume $N > 2k$ and $k \ge 3$. The actual scheme chooses a large $N$ and $k=3$. 
\def\seq{{\sf seq}}

\paragraph{Insecurity of the scheme}
We note that the access pattern of a single ORAMWrite is hidden. However, that
alone is not sufficient to show write obliviousness. In particular, security
breaks down when one considers {\em multiple ORAMWrite operations}.

Observe that the above algorithm is more likely to choose a free
block than a non-free block; with probability 1/2, a chosen block will be from
$S_0$ and thereby always free, and with probability 1/2, a chosen block will be
from $S_1$ and thereby sometimes free. This tendency towards choosing free
blocks leaks information. To clarify our point, consider the following two
sequences of logical writes:
$$\seq^0 = (init, w_0, w_0, w_2), ~~~~~~ 
\seq^1 = (init, w_0, w_1, w_2)$$
Here, $w_i$ denotes writing data to a logical address $i$, and $init$ is a
sequence of operations that makes the ORAM have exactly $N$ free
blocks.\footnote{
Their ORAM seems to be initialized with exactly $N$ free blocks, in which
case $init$ contains no operation. If that's not the case, we can set the
$init$ sequence as follows: $$init = (w_0, \ldots, w_{N-1}, \underbrace{w_0,
\ldots, w_0}_{\mbox{$\secparam$ times}}), $$
where $\secparam$ is the security parameter. 
Note that after the $init$ sequence, the ORAM will have exactly $N$ non-free
physical blocks and $N$ free physical blocks with probability least $1 -
\negl(\secparam)$. So, we can safely ignore this negligible probability,
and proceed our argument assuming that the ORAM has exactly $N$ free blocks
after the $init$ sequence.} 

Let $U_i = (u_{i,1}, \ldots, u_{i,k})$ be the set of chosen blocks from the
$i$th ORAMWrite after the $init$ sequence. Let $d_\ell$ be the data in {\em logical block} $\ell$. 

Then, in $\seq^0$, physical block $\gamma \in U_1$ containing $d_0$ will be
probably freed up thanks to the second $w_0$, and the last $w_2$ may be
able to choose $\gamma$ as a free block. However, in $\seq^1$, the block
$\gamma$ cannot be freed up by $w_1$, since $\gamma$ contains $d_0$! So,
the last $w_2$ can choose $\gamma$ only as a non-free block. Due to the
different probablity weights in choosing free blocks vs.~non-free blocks, $U_1$
and $U_3$ are more likely to overlap in $\seq^0$ than in $\seq^1$, and security
breaks down.

To clarify our point, we give an attack. Given an access pattern $(U_1, U_2,
U_3)$, the adversary tries to tell if it is from $\seq^0$ or
$\seq^1$.
Consider the following events: 
\BI
\item $X$: $u_{1,1} \not \in U_2$, $~~~~~$ 
$Y$: $u_{1,1} \in U_3$, $~~~~~$  $E$: $X \wedge Y$
\EI

\begin{framed}
\noindent The adversary works as follows: \\
\phantom{...} Output 0 if $E$ takes place; otherwise output a random bit.
\end{framed}

Let $p^b = \Pr[E]$ from $\seq^b$. We show that $p^0 - p^1$ is
non-negligible, which proves that the adversary is a good distinguisher. 

Let $F_i(u)$ denote a predicate indicating whether a physical block $u$ was free
when the $i$th ORAMWrite starts.
Note that whether $u_{1,1}$ belongs $U_3$ ultimately depends on $F_3(u_{1,1})$. In particular,
for any $u_{1,1}$, we have 
\begin{eqnarray*}
q_y & = & Pr[Y ~|~ F_3^b(u_{1,1})] = \frac 1 2 \cdot \frac k N + \frac 1 2 \cdot \frac k {2N - k}  \\
q_n & = & Pr[Y ~|~ \neg F_3^b(u_{1,1})] = \frac 1 2 \cdot \frac k {2N - k} 
\end{eqnarray*}

\def\FF{{\sf FreeSet}}
The following table shows how $F_3(u_{1,1})$ depends on the previous events. 
In the table, $D_2(u) \in \{f, d_0^*, \ldots, d_{N-1}^*, d_0 \}$
denotes a random variable indicating which logical block a physical block $u$
contains when the second ORAMWrite starts. If the value is $f$, it means the
block is free, and $d_\ell^*$ is the initial data for the logical block $\ell$
that the ORAM initialization procedure used. The value $d_0$ denotes the data
block used in the first $w_0$ operation in $\seq^0$ and $\seq^1$. 
Let $S_i(\ell)$ denote a predicate indicating whether a logical block $\ell$ is
in the stash when the $i$th ORAMWrite starts.
In addition, $\FF_i$ denotes a predicate indicating whether the $i$th ORAMWrite
found a physical block in the free set $S_0$ (thereby successfully writing the
input logical block in the free physical block). 

\medskip

\noindent
\begin{tabular}{c | c c c || c c} 
\hline
case & $D_2(u_{1,1})$ & $S_2(d_0)$ 
& \parbox{5em}{$\FF_2$ \\(cond.~on $X$)} 
& \parbox{3.5em}{$F_3(u_{1,1})$ \\ ($\seq^0$)} 
& \parbox{3.5em}{$F_3(u_{1,1})$ \\ ($\seq^1$)} \\
\hline
$c1$ & $f$  & x & x   & 1 & 1 \\ 
$c2$ & $d_0$   & 0    & 0 & 0 & 0 \\
\fbox{$c3$} & $d_0$   & 0    & 1 & 1 & 0 \\
$c4$ & $d_{\ge 0}^*$   & x    & 0 & 0 & 0 \\
$c5$ & $d_0^*$   & 1    & 1 & 1 & 1 \\
\fbox{$c6$} & $d_1^*$   & x    & 1 & 0 & 0 or 1 \\
$c7$ & $d^*_{\ge 2}$ & x & 1 & 0 & 0 \\ 
\hline
\end{tabular}

\medskip
For example, in case $c3$, 
\BI
\item $D_2(u_{1,1}) = d_0$: When the second ORAMWrite begins, the physical block $u_{1,1}$ contains the logical block $d_0$. 
\item $S_2(d_0) = 0$: When the second ORAMWrite begins, the stash is empty.
\item $(\FF_2 | X) = 1$: The second ORAMWrite found at least one block in the free set $S_0$.
\item For $\seq^0$, the second write is $w_0$. From $\FF_2 = 1$, a new
$d_0$ from $w_0$ will be written in a free block, and $u_{1,1}$ containing
the old $d_0$ is freed.  

\item For $\seq^1$, the second write is $w_1$.
From $\FF_2 = 1$, a new $d_1$ from $w_1$ will be written in a free block,
but $u_{1,1}$ containing $d_0$ is not affected.
\EI

\medskip
Note that the first ORAMWrite in both $\seq^0$ and $\seq^1$ is the same with
$w_0$, so $D_2(u_{1,1})$ and $S_2(d_0)$ is identically distributed for both
$\seq^0$ and $\seq^1$. Moreover, observe that the distribution of $X$ depends
on only $u_{1,1}$ because the ORAMWrite samples $U_2$ at random. Finally,
$\Pr[\FF_2]$ is always the same, since the number of free blocks in the
second ORAMWrite is always the same with $N$.  

Based on the table and the above observation, we have the following:

\begin{eqnarray*}
p^0 - p^1 & \ge & \Pr[c3 \wedge X](q_y - q_n) - \Pr[c6 \wedge X](q_y - q_n)  \\
& \ge & \frac {k} {2N} \cdot \big( \Pr[c3 \wedge X] - \Pr[c6 \wedge X] \big) \\
& = & \frac {k} {2N} \cdot \Pr[X] \cdot \big( \Pr[c3 | X] - \Pr[c6 | X] \big) \\
& \ge & \frac {k} {4N} \cdot \big( \Pr[c3 | X] - \Pr[c6 | X] \big) \\
\end{eqnarray*}

Now, let's first calculate the lower bound on $\Pr[c3 | X]$.  If the first ORAMWrite
chooses at least one block from the free set and
writes $d_0$ in $u_{1,1}$, it must be $D_2(u_{1,1}) = d_0$ and $S_2(d_0)=0$.  
Therefore, 
$$\Pr[D_2(u_{1,1}) = d_0, S_2(d_0) = 0] \ge \frac 1 2 \cdot \frac 1 k.$$
Moreover, at least probability $\frac 1 2$, the second ORAMWrite will find a
block from the freeset, which implies that $$\Pr[c3 | X ] \ge \frac 1 {4k}.$$

To calculate the upper bound on $\Pr[c6 | X]$, observe that $D_2(u_{1,1}) = d^*_1$
implies that $u_{1,1}$ contained $d^*_1$ even before the first
ORAMWrite $w_0$.  Therefore, we have 

$$\Pr[c6 | X ] \le \Pr[\mbox{ $u_{1,1}$ has $d^*_1$ before the 1st ORAMWrite}] = \frac 1 {2N}.$$

Therefore, we have
$$p^0 - p^1 \ge \frac {k} {4N} \cdot \Big ( \frac 1 {4k} - \frac 1 {2N} \Big) = \frac {N-2k} {4N^2}.$$

\section{Related Work}
\label{sec:related}

\paragraph{Oblivious RAM (ORAM) and applications.}
ORAM protects the access pattern so that it is infeasible to guess which
operation is occurring and on which item. Since the seminal work by Goldreich
and Ostrovsky~\cite{GO96}, many works have focused on improving efficiency and
security of ORAM (for example
\cite{CCS:SDSFRY13,RFK+15,CCS:MoaMayBla15,SP:SZALT16} just to
name a few; see the references therein).
 
ORAM plays as an important tool to achieve secure cloud
storage~\cite{NDSS:SteShi13,SP:SteShi13,NDSS:MayBlaCha14} and secure
multi-party
computation~\cite{CCS:GKKKMR12,SP:LHSKH14,SP:LWNHS15,CCS:WanChaShi15,SP:ZWRGDE16}
and secure processors~\cite{freecursive,ghostrider,hop}.  There also have been
works to hide the access pattern of protocols accessing individual data
structures, e.g., maps, priority queues, stacks, and queues and graph
algorithms on the cloud
server~\cite{PODC:Toft11,ASIACCS:BlaSteAli13,CCS:WNLCSS14,SP:RocAviCho16}.  The
work of \cite{JMTS16} considers obliviousness in the P2P content sharing
system. 

\paragraph{Write-only obliviousness}
Blass~et~al.~\cite{CCS:BMNO14} considers write-only ORAM (WoORAM), and gave a
WoORAM construction much more efficient than the traditional ORAM
constructions. They applied WoORAM to deniable storage scenarios and gave a
WoORAM-based construction of hidden volume encryption (HiVE). 
Aviv~et~al.~\cite{oblivisync} gave a construction of oblivious synchronization
and backup for the cloud environment. They observed that write-only
obliviousness is sufficient for the scenario, since the client stores a
complete local copy of his data, and therefore $read$ accesses are naturally
hidden from the adversary.

\paragraph{Deniable storage}
Anderson et al.~\cite{ANS98} proposed steganography-based approaches, that is,
hiding blocks within cover files or random data. There are
works based on his suggestion~\cite{StegFS,StegFS2}, but they don't
allow deniability against a snapshot adversary.

Another approach is hidden volumes. Unfortunately, existing solutions such
as TrueCrypt (discontinued now)~\cite{truecrypt}, Mobiflage~\cite{mobiflage}
and MobiPluto~\cite{mobipluto} are secure only against a single-snapshot
adversary. 
HIVE~\cite{CCS:BMNO14} provides security even against a multiple-snapshot
adversary.
DEFY~\cite{NDSS:PetGonPet15} is the deniable log-structured file system
specifically designed for flash-based, solidstate drives; although it is secure
against a multiple-snapshot adversary, it doesn't scale well.

\section{Conclusion}
We presented \ourworam{}, a stash-free, deterministic write-only
oblivious ORAM with sequential write patterns. This scheme achieves
asymptotic improvement in write communication costs, $O(B\log N)$,
requiring exactly $2B$ physical writes per logical write. We further
showed that prior schemes to improve on \travisworam{} are
insecure. Finally, we implemented and evaluated \ourworam{}, and, for
sequential writing, it incurs only a 3-4.5x overhead on HDD and 10-14x
on SSD compared to using encryption only. It is 19x faster than
\travisworam{}, the previous best, secure scheme. It is also
practical; the theoretical write complexity is optimal, and
\ourworam{} with an SSD backend has similar (sometimes better)
performance compared to using
encryption only on a spinning-platter HDD in a similar price range.

\begin{acks}
  The authors thank the CCS program committee for their valuable
  suggestions.
  This work is supported by the
  \grantsponsor{onr}{Office of Naval Research}{https://www.onr.navy.mil/}
  under awards
  \grantnum{onr}{N0001416WX01489} and \grantnum{onr}{N0001416WX01645},
  and by the
  \grantsponsor{nsf}{National Science Foundation}{https://www.nsf.gov/}
  under awards
  \grantnum{nsf}{1618269} and
  \grantnum{nsf}{1319994}.
\end{acks}

\bibliographystyle{ACM-Reference-Format}
\balance


\end{document}